\documentclass[journal,10pt]{IEEEtran}
\renewcommand{\baselinestretch}{0.99}

\usepackage{hyperref}
\usepackage{cite}
\usepackage{amsmath,amsthm,amssymb,amsfonts,mathtools}
\usepackage{graphicx,flushend}
\usepackage{algorithm,algorithmic}
\usepackage{subcaption}
\usepackage{booktabs}
\usepackage{xcolor}
\usepackage{comment}

\def\BibTeX{{\rm B\kern-.05em{\sc i\kern-.025em b}\kern-.08em
    T\kern-.1667em\lower.7ex\hbox{E}\kern-.125emX}}

\newtheorem{prop}{Proposition}

\DeclareMathOperator{\Tr}{Tr}

\DeclareMathOperator{\PEB}{PEB}
\DeclareMathOperator{\CRB}{CRB}

\DeclareMathOperator{\Rea}{Re}
\DeclareMathOperator{\Ima}{Im}
\DeclareMathOperator{\atanTwo}{atan2}

\newcommand{\ee}{\mathrm{e}}
\newcommand{\norm}[1]{\left\lVert #1 \right\rVert}
\newcommand{\abs}[1]{\left\lvert #1 \right\rvert}
\newcommand{\br}[1]{\left( #1 \right)}
\newcommand{\sbr}[1]{\left[ #1 \right]}

\newcommand{\p}{\mathbf{p}}
\newcommand{\q}{\mathbf{q}}

\newcommand{\s}{\mathbf{s}}

\newcommand{\y}{\mathbf{y}}
\newcommand{\z}{\mathbf{z}}

\newcommand{\I}{\mathbf{I}}
\newcommand{\J}{\mathbf{J}}

\newcommand{\R}{\mathbf{R}}
\renewcommand{\S}{\mathbf{S}}
\newcommand{\T}{\mathbf{T}}

\newcommand{\Compl}{\mbox{$\mathbb{C}$}}

\newcommand{\diag}{\mathrm{diag}}

\begin{document}

\title{Wideband Sensing with Dynamic Metasurface Antennas under Realistic Phase Response Modeling}

\author{\IEEEauthorblockN{
Ioannis Gavras$^1$ and George C. Alexandropoulos$^{1,2}$ 
} 
\\
\IEEEauthorblockA{$^1$Department of Informatics and Telecommunications, National and Kapodistrian University of Athens, Greece}
\IEEEauthorblockA{$^2$Department of Electrical and Computer Engineering, University of Illinois Chicago, USA}
\\
\IEEEauthorblockA{emails: \{giannisgav, alexandg\}@di.uoa.gr}
\vspace{-1cm}
\thanks{This work has been supported by the SNS JU project 6G-DISAC under the EU's Horizon Europe research and innovation programme under Grant Agreement number 101139130.
}}

\maketitle

\begin{abstract}
This paper investigates the impact of practical features of  the emerging antenna array technology of Dynamic Metasurface Antennas (DMAs) when used for wideband sensing. By adopting a realistic DMA response model capturing frequency selective magnetic polarizability, finite resonant frequency tuning, and waveguide phase and leakage effects, we first present a compact observation model for user localization and multiple scattering points sensing through DMA-based analog combining of Orthogonal Frequency Division Multiplexing (OFDM) pilots transmitted in the uplink direction. Building on this model, we derive the Fisher Information Matrix (FIM), the equivalent FIM, and the corresponding Cram\'er--Rao Bounds (CRBs) for the relevant spatitemporal parameters estimation. The analysis reveals that frequency selectivity reduces the effective information bandwidth and distorts the DMA-based reception manifold, while waveguide attenuation decreases both the coherent combining gain and the effective aperture, thereby degrading estimation accuracy. Numerical results validate the analysis and confirm the resulting inflation in the delay, angle, and position error bounds.

\end{abstract}

\vspace{-0.1cm}

\section{Introduction}
Localization and environmental awareness constitute core internal functionalities as well as end services of the upcoming $6$-th Generation (6G) of wireless networks~\cite{6G-DISAC_mag}. In this context, Dynamic Metasurface Antennas (DMAs) provide a promising scalable antenna array architecture for compact and energy-efficient wide-aperture sensing~\cite{Shlezinger2021Dynamic}, owing to their waveguide-fed structure and electronically tunable metamaterial elements~\cite{35,gavriilidis2025microstrip}. However, unlike ideal phased arrays, practical DMAs operate under hardware-induced constraints, including resonant frequency selectivity, finite tuning range, and waveguide attenuation \cite{carlson2025wideband}. These characteristics become particularly pronounced in wideband operation, where they can substantially alter the sensing information \cite{gavras2024circuit,gavras2025emcompliant_dma_bistatic}.

Prior works have developed waveguide-fed response-tunable metasurface elements and demonstrated electronically reconfigurable DMA apertures \cite{38,35,jabbar202460}, while DMA-based transceivers have been studied for uplink massive Multiple-Input Multiple-Output (MIMO), multiuser MIMO, MIM with Orthogonal Frequency Division Multiplexing (OFDM), and hybrid analog and digital precoding with limited Radio Frequency (RF) chains \cite{36,Shlezinger2021Dynamic,YXA2023,gavriilidis2025nearfield}. DMAs have also recently emerged in sensing, localization, and near-field Integrated Sensing and Communications (ISAC) settings~\cite{gavras2025nearfield_dma_crb,smida2024inband}. Nevertheless, most existing DMA signal processing models remain narrowband and represent the aperture only through Lorentzian-constrained weights, thus, neglecting practical effects such as frequency selectivity, finite tuning range, and waveguide leakage. While the authors in~\cite{carlson2025wideband} recently introduced a realistic wideband DMA model for communication-oriented beamforming and demonstrated substantial deviations from narrowband Lorentzian abstractions at large bandwidths, the impact of such practical DMA effects on sensing and localization accuracy remains insufficiently understood.

Motivated by the latter research gap, this paper studies the impact of realistic wideband DMA characteristics on sensing when considered for a scalable multi-antenna Receiver (RX). We consider an uplink scenario in which a User Equipment (UE) transmits OFDM pilot symbols that are received by the DMA-based RX through a direct path and multiple reflected paths generated by nearby Scattering Points (SPs). To this end, inspired by~\cite{carlson2025wideband}, we consider a practical wideband DMA reception model that accounts for multiple RF outputs and multiple phase configurations, and formulate the corresponding path- and geometry-domain models for direct and reflected uplink propagation. Building on this framework, we derive the Fisher Information Matrix (FIM) and the associated Cram\'er--Rao Bounds (CRBs) for the relevant position parameters. The analysis reveals that frequency selectivity constrains the effective information bandwidth and alters the receive manifold, whereas waveguide attenuation lowers both the coherent combining gain and the effective aperture, thereby reducing localization accuracy. The presented numerical results validate these findings by confirming the resulting inflation in the delay, angle, and Position Error Bounds (PEBs).

\section{System and Channel Models}\label{Sec:System}

We consider a multicarrier uplink sensing system in which the RX employs a DMA with $N_{\rm RF}$ microstrips, each connected to a dedicated reception RF chain. Each microstrip supports $N_{\rm E}$ response-tunable metamaterials, yielding a total of $N \triangleq N_{\rm RF}N_{\rm E}$ antenna elements, arranged as a Uniform Linear Array (ULA). The RX is tasked with sensing the UE and $S$ targets, modeled as SPs, all located in its vicinity. The DMA metamaterial elements realize low power analog combining via adjustable resonant responses~\cite{carlson2025wideband}. The sensing process is enabled by $T$ known OFDM pilot symbols transmitted by the UE over $K$ SubCarriers (SCs) within each channel coherence block, while the RX cycles through $J$ phase configurations.

Let $f_c$ denote the carrier frequency, $\Delta f$ the SC spacing, and $ f_k \triangleq f_c+\left(k-\frac{K+1}{2}\right)\Delta f$, with $ k=1,\ldots,K,$
be the $k$-th SC frequency bin. The pilot matrix is denoted by $\mathbf S\in\Compl^{K\times T}$, whose $k$-th row is $\s_k\triangleq[\S]_{:,k}\in\Compl^{1\times T}$ and satisfying $\mathbb{E}\{\|\mathbf s_k\|^2\}=T$. For each $j$-th DMA phase configuration ($j=1,\ldots,J$), corresponding to a different set of resonant frequencies~\cite{gavriilidis2025nearfield}, on the $k$-th SC, the RX applies the analog combining operation represented by $\mathbf W_j[k]\in\Compl^{N\times N_{\rm RF}}$.

\subsection{Receive Signal Model}
Let $\mathbf r_{j}[k]\in\Compl^{N\times T}$ denote the received signal corresponding to the $j$-th DMA configuration at the $k$-th SC. Assuming a total of $L\triangleq S+1$ propagation paths, comprising one direct path ($\ell=0$) and $S$ reflected paths ($\ell=1,\ldots,S$), the $N$-element received signal in baseband representation is expressed as:
\begin{align}
    \mathbf r_{j}[k]
    =
    \sqrt{P}\sum_{\ell=0}^{L-1}
    \gamma_\ell
    \ee^{-\jmath 2\pi f_k \tau_\ell}
    \mathbf a_\ell[k]\mathbf s_k
    +
    \mathbf n_{j}[k],
    \label{eq:rjk}
\end{align}
where $P$ is the UE transmit power, while $\gamma_\ell$, $\tau_\ell$, and $\mathbf a_\ell[k]\in\Compl^{N\times1}$ denote the complex gain, propagation delay, and array response associated with the $\ell$-th path, respectively. Moreover, $\mathbf n_{j}[k]$ represents Additive White Gaussian Noise (AWGN) with independent entries distributed according to $\mathcal{CN}(0,\sigma^2)$.

After analog combining and pilot removal, the resulting $N_{\rm RF}$-dimensional vector available for baseband processing at the RX side is given as follows:
\begin{align}
    \mathbf y_{j}[k]
    &\triangleq\nonumber
    \frac{1}{\sqrt{T}}\mathbf W_j^{\rm H}[k]\mathbf r_{j}[k]\mathbf s_k^{\rm H}\\
    &=
    \sqrt{P}\sum_{\ell=0}^{L-1}
    \gamma_\ell
    \ee^{-\jmath 2\pi f_k \tau_\ell}
    \mathbf g_{j,\ell}[k]
    +
    \mathbf z_{j}[k],
    \label{eq:yjk}
\end{align}
where $\mathbf g_{j,\ell}[k]\triangleq \mathbf W_j^{\rm H}[k]\mathbf a_\ell[k]\in\Compl^{N_{\rm RF}\times1}$ is the effective DMA sensing manifold of the $\ell$-th path at the $j$-th DMA configuration and the $k$-th SC, and $\z_{j,k}$ denotes the effective post-combining noise vector, which is distributed as $\mathbf z_{j}[k]\sim \mathcal{CN}\!\br{\mathbf 0,\sigma^2\mathbf I_{N_{\rm RF}}}$ following from the independently distributed AWGN entries in $\mathbf n_{j}[k]$ and the normalized combining process, i.e., $\mathbf W_j^{\rm H}[k]\mathbf W_j[k]\cong \mathbf I_{N_{\rm RF}}$. Stacking $\y_{j}[k]$'s across all $J$ DMA phase configurations and all $K$ SCs, yields the compact observation vector $\widetilde{\mathbf{y}}\triangleq    \sbr{\mathbf y_{1}^{\rm T}[1],\ldots,\mathbf y_{1}^{\rm T}[K],\mathbf y_{2}^{\rm T}[1],\ldots,\mathbf y_{J}^{\rm T}[K]}^{\rm T}
    \in\Compl^{JKN_{\rm RF}\times1}.$

\vspace{-0.2cm}
\subsection{DMA Analog Combining Model}
Following \cite{carlson2025wideband}, each DMA element response can modeled using the frequency selective magnetic polarizability law as:
\begin{align}
    \alpha_M(f,f_r)
    \triangleq
    \frac{2\pi f^2F_{\rm cpl}}
    {2\pi f_r^2-2\pi f^2+\jmath\Gamma f},
    \label{eq:alphaM}
\end{align}
where $f_r$ denotes the element's resonant frequency, $F_{\rm coupl}$ is the coupling coefficient, and $\Gamma$ is the damping factor. Using $Q(f)\triangleq2\pi f/\Gamma$, the normalized DMA response becomes:
\begin{align}
    \bar{\alpha}_M(f,f_r)=\frac{1}{Q(f)F}\alpha_M(f,f_r).
\end{align}
An equivalent Lorentzian-form representation is $\bar{\alpha}_M(f,f_r)=0.5(-\jmath-\ee^{\jmath2\Psi(f,f_r)})$, where $\Psi(f,f_r)\triangleq\arctan\!\br{\frac{2\pi(f_r^2-f^2)}{\Gamma f}}$.

To account for waveguide propagation effects, let $\ell_{n,m}$ denote the propagation distance from the $n$-th metamaterial ($n=1,\ldots,N$) to the input of the $m$-th RF chain ($m=1,\ldots,N_{\rm RF}$) along the associated waveguide path, and let $\beta_{g,m}(f)$ and $\bar{\alpha}_{g,m}(f)$ denote the corresponding waveguide phase and attenuation constants, respectively \cite{carlson2025wideband}. Then, each $(n,m)$-th entry of the DMA-based analog combining operation, represented by the matrix $\mathbf W_j[k]$, can be expressed as:
\begin{align}
    &\nonumber[\mathbf W_j[k]]_{n,m}
    =\\
    &c_{n,m}
    \ee^{-(\bar{\alpha}_{g,m}(f_k)+\jmath\beta_{g,m}(f_k))\ell_{n,m}}
    \bar{\alpha}_M\!\br{f_k,f_{r,n,m}^{(j)}},
    \label{eq:Wjk}
\end{align}
where $c_{n,m}$ is a fixed coupling coefficient determined by the hardware interconnection between the $n$-th element of the $m$-th microstrip and its attached RF chain~\cite{carlson2025wideband}, and $f_{r,n,m}^{(j)}$ denotes the corresponding resonant frequency of the $j$-th element configuration. In addition, $\beta_{g,m}(f_k)$ and $\bar{\alpha}_{g,m}(f_k)$ defined in~\cite[eqs. (7)-(9)]{carlson2025wideband} represent the waveguide phase and attenuation constants, respectively, at each $k$-th frequency bin $f_k$. Expression~\eqref{eq:Wjk} clearly highlights that the DMA receive manifold is inherently wideband and dispersive due to two distinct mechanisms: \textit{i}) the element-level resonance response $\bar{\alpha}_M(f_k,f_r)$; and \textit{ii}) the waveguide frequency responses $\beta_{g,m}(f_k)$ and $\bar{\alpha}_{g,m}(f_k)$. Finally, the resonant frequencies are constrained as $f_{r,\min}\leq f_{r,n,m}^{(j)}\leq f_{r,\max}$ $\forall n,m,j$, where $f_{r,\min}$ and $f_{r,\max}$ denote the lower and upper tuning limits, respectively; the tuning bandwidth is defined as $B_{\rm tune}\triangleq f_{r,\max}-f_{r,\min}$.

\vspace{-0.2cm}
\subsection{System Geometry}
We adopt a two-dimensional Cartesian geometry in which the DMA-based RX is located at the origin, the UE is positioned at $\p_{\rm U}\triangleq[x_{\rm U},y_{\rm U}]$, and the $s$-th target ($s=1,\ldots,S$) is located at $\q_s\triangleq[x_s,y_s]$. Under the far-field propagation assumption, the steering vector associated with the $\ell$-th path at the $k$-th SC is given by $\mathbf a_\ell[k] =\frac{1}{\sqrt{N}}\sbr{1,\ldots,\ee^{\jmath \frac{2\pi f_k}{c}(N-1)d_x\sin\phi_\ell}}^{\rm T},$ where $d_x$ denotes the element spacing and $\phi_\ell$ is the corresponding Angle of Arrival (AoA). The path delays can be then expressed as follows:
\begin{align}
    \tau_\ell
    \triangleq
    \begin{cases}
        \frac{\norm{\mathbf p_{\rm U}-\mathbf p_{\rm D}}}{c}, & \ell=0\\
        \frac{\norm{\mathbf p_{\rm U}-\mathbf q_\ell}+\norm{\mathbf q_\ell-\mathbf p_{\rm D}}}{c}, & \ell=1,\ldots,S
    \end{cases},
\end{align}
with $c$ denoting the speed of light, whereas the corresponding AoAs are expressed as:
\begin{align}
    \phi_\ell
    \triangleq
    \begin{cases}
        \atanTwo\!\br{y_{\rm U}-y_{\rm D},x_{\rm U}-x_{\rm D}}, & \ell=0\\
        \atanTwo\!\br{y_\ell-y_{\rm D},x_\ell-x_{\rm D}}, & \ell=1,\ldots,S
    \end{cases}.
\end{align}
Finally, the path gains $\gamma_\ell$ account for both propagation attenuation and reflection coefficients.

\section{Estimation Bounds and CRB Characterization}\label{Sec:CRB}
In this section, we derive the FIM and the corresponding CRBs for the considered DMA sensing model. We begin with a path-domain parametrization of the observation vector, and then map the results to position-domain bounds.

\vspace{-0.2cm}
\subsection{FIM Analysis}
We first collect the relevant channel parameters into the path-domain parameter vector $\boldsymbol{\eta}\triangleq[\boldsymbol{\tau}^{\rm T},\boldsymbol{\phi}^{\rm T},\boldsymbol{\gamma}_{\rm R}^{\rm T},\boldsymbol{\gamma}_{\rm I}^{\rm T}]^{\rm T}
    \in\R^{4L\times1}$, where $\boldsymbol{\tau}\triangleq \sbr{\tau_0,\ldots,\tau_{L-1}}^{\rm T}\in\mathbb{R}^{L\times 1}$, $\boldsymbol{\phi} \triangleq \sbr{\phi_0,\ldots,\phi_{L-1}}^{\rm T}\in\mathbb{R}^{L\times 1}$, $ \boldsymbol{\gamma}_{\rm R} \triangleq \sbr{\Rea\{\gamma_0\},\ldots,\Rea\{\gamma_{L-1}\}}^{\rm T}\in\mathbb{R}^{L\times 1}$, and $\boldsymbol{\gamma}_{\rm I} \triangleq \sbr{\Ima\{\gamma_0\},\ldots,\Ima\{\gamma_{L-1}\}}^{\rm T}\in\mathbb{R}^{L\times 1}$. From \eqref{eq:yjk}, the observation vector $\mathbf y_{j}[k]$ is circularly symmetric complex Gaussian with mean $\boldsymbol{\mu}_{j,k}\triangleq\mathbb{E}[\y_{j}[k]]=\sqrt{P}
    \sum_{\ell=0}^{L-1}\gamma_\ell\ee^{-\jmath2\pi f_k\tau_\ell}\mathbf g_{j,\ell}[k]$ and covariance $\mathbb{E}[(\y_{j}[k]-\boldsymbol{\mu}_{j,k})(\y_{j}[k]-\boldsymbol{\mu}_{j,k})^{\rm H}]=\sigma^2\I_{\rm N_{\rm RF}}$. Hence, the FIM with respect to $\boldsymbol{\eta}$ is given by $\mathbf J_{\eta}\triangleq\sum_{j=1}^{J}\sum_{k=1}^{K}\mathbf J_{j,k},$ where the $(u,v)$-th entry of the local FIM $\mathbf J_{j,k}$, with $u,v=1,\ldots,4L$, becomes:
\begin{align}
    \sbr{\mathbf J_{j,k}}_{u,v}
    \triangleq
    \frac{2P}{\sigma^2}
    \Rea\!\left\{
    \frac{\partial \boldsymbol{\mu}_{j,k}^{\rm H}}{\partial [\boldsymbol{\eta}]_u}
    \frac{\partial \boldsymbol{\mu}_{j,k}}{\partial [\boldsymbol{\eta}]_v}
    \right\}.
    \label{eq:Jjk}
\end{align}

Let the position-domain parameter vector be  $\boldsymbol{\xi}\triangleq\sbr{\mathbf p_{\rm U}^{\rm T},\mathbf q_1^{\rm T},\ldots,\mathbf q_S^{\rm T},\boldsymbol{\gamma}_{\rm R}^{\rm T},\boldsymbol{\gamma}_{\rm I}^{\rm T}}^{\rm T}\in\mathbb{R}^{2(S+L+1)\times 1}$. The path- and geometry-domain parameters are related through the Jacobian transformation matrix $\mathbf{T}$, whose each $(u,v)$-th entry is $ [\mathbf T]_{u,v}\triangleq \frac{\partial[\boldsymbol{\eta}]_u}{\partial[\boldsymbol{\xi}]_v}$. Therefore, the geometry-domain FIM can be expressed as $\widetilde{\mathbf J}=\mathbf{T}^{\rm T}\J_{\eta}\mathbf{T}.$ Next, we perform the partition $\boldsymbol{\xi}=[\boldsymbol{\xi}_i^{\rm T},\boldsymbol{\xi}_n^{\rm T}]^{\rm T}$, where $\boldsymbol{\xi}_i$ contains the position parameters of interest and $\boldsymbol{\xi}_n$ contains nuisance path gains. It can be easily concluded that the geometry-domain FIM $\widetilde{\J}$ admits the following block partition:
\begin{align}
    \widetilde{\mathbf J}
    =
    \begin{bmatrix}
        \widetilde{\mathbf J}_{ii} & \widetilde{\mathbf J}_{in}\\
        \widetilde{\mathbf J}_{in}^{\T} & \widetilde{\mathbf J}_{nn}
    \end{bmatrix}.
\end{align}
Using the Schur complement, the Equivalent FIM (EFIM) for the parameters of interest deduces to $\widetilde{\mathbf J}_{\rm e}\triangleq\widetilde{\mathbf J}_{ii}-\widetilde{\mathbf J}_{in}\widetilde{\mathbf J}_{nn}^{-1}\widetilde{\mathbf J}_{in}^{\rm T}.$ The PEB for the UE and the SPs is then obtained as:
\begin{align}
    \PEB
    \triangleq
    \sqrt{
    \Tr\!\left\{
    [\widetilde{\mathbf J}_{\rm e}^{-1}
    ]_{1:2(S+1),1:2(S+1)}\right\}
    }.
    \label{eq:PEB}
\end{align}

\section{DMA Implications on Multicarrier Sensing}\label{Sec:Impact}
This section investigates the effect of the two practical DMA mechanisms that are most relevant for sensing: frequency selectivity and waveguide attenuation. The analysis proceeds directly from the aforedescribed FIM/EFIM characterization, and then provides interpretable approximations.

\vspace{-0.2cm}
\subsection{Delay Information and Effective Information Bandwidth}
For each single propagation path (corresponding to either the UE or an SP), after omitting the path index $\ell$ for notational convenience, the delay-related term in \eqref{eq:Jjk} is given by:
\begin{align}\label{eq:Jtautau0}
    J_{\tau\tau}=4\pi^2\sum_{j=1}^{J}\sum_{k=1}^{K}\omega_{j,k}f_k^2,
\end{align}
where we have used the definition $\omega_{j,k}\triangleq\frac{2P|\gamma|^2}{\sigma^2}\norm{\mathbf g_j[k]}^2$.

\begin{prop}
Let $\bar{f}_{\omega}\triangleq\frac{\sum_{j=1}^J\sum_{k=1}^K\omega_{j,k}f_k}{\sum_{j=1}^J\sum_{k=1}^K\omega_{j,k}}$ denote the weighted mean frequency, and define $\beta_{\rm eff}^2\triangleq\frac{\sum_{j=1}^J\sum_{k=1}^K\omega_{j,k}(f_k-\bar{f}_{\omega})^2}{\sum_{j=1}^J\sum_{k=1}^K\omega_{j,k}}$ as the corresponding centered frequency spread. Then, the delay-related FIM term in expression~\eqref{eq:Jtautau0} can be expressed as follows:
\begin{align}
    J_{\tau\tau}
    =
    4\pi^2(\beta_{\rm eff}^2+\bar f_{\omega}^2)\sum_{j=1}^{J}\sum_{k=1}^{K}\omega_{j,k}.
    \label{eq:decomp_jtau}
\end{align}
Upon elimination of the nuisance complex path coefficients through the EFIM, the delay information is governed exclusively by the term $\beta_{\rm eff}^2$.
\end{prop}

\begin{proof}
The result follows by constructing the FIM for the delay jointly with the unknown complex path coefficient, and subsequently eliminating the latter via the Schur complement. The component involving the weighted mean SC frequency induces only a common phase rotation across SCs, which is indistinguishable from a modification of the unknown complex gain, and is therefore non-identifiable. After this nuisance parameter is removed, the only remaining source of delay information is the spread of the SC frequencies around their weighted mean. This term is exactly the weighted variance, yielding the claimed expression.
\end{proof}

Proposition~1 indicates that the impact of frequency selectivity on the delay estimation is captured by the weights $\omega_{j,k}$'s. Specifically, if the DMA suppresses edge SCs owing to the dispersive response of its elements or to beam-squint mismatch, the effective squared bandwidth $\beta_{\rm eff}^2$ decreases, thereby leading to a larger CRB performance.

\vspace{-0.2cm}
\subsection{Angular Information and Effective Aperture}
It follows from \eqref{eq:Jjk} that the AoA-related Fisher information term, with the path index $\ell$ omitted for simplicity, is given by:
\begin{align}
    J_{\phi\phi}
    =
    \frac{2P|\gamma|^2}{\sigma^2}
    \sum_{j=1}^{J}\sum_{k=1}^{K}
    \left\|
    \mathbf W_j^{\rm H}[k]
    \frac{\partial\mathbf a[k]}{\partial \phi}
    \right\|^2.
    \label{eq:Jphiphi0}
\end{align}
Let us now define the metamaterial element power profile induced by the DMA analog combiner as follows:
\begin{align}
    p_{j,n}[k]
    \triangleq
    \frac{\sum_{m=1}^{N_{\rm RF}}\abs{\sbr{\mathbf W_j[k]}_{n,m}}^2}
    {\sum_{q=1}^{N}\sum_{m=1}^{N_{\rm RF}}\abs{\sbr{\mathbf W_j[k]}_{q,m}}^2},
\end{align}
and its centroid \(\bar{x}_{j,k}\triangleq\sum_{n=1}^{N}p_{j,n}[k]x_n\), where \(x_n=(n-1)d_x\)  \(\forall n\). Then, we define the effective aperture moment as:
\begin{align}
    D_{{\rm eff},j}^2[k]
    \triangleq
    \sum_{n=1}^{N}p_{j,n}[k]\left(x_n-\bar{x}_{j,k}\right)^2.
    \label{eq:Deff}
\end{align}

\begin{prop}
Under matched combining, i.e., when the DMA weights are aligned with the local incident field, the angular-information term in \eqref{eq:Jphiphi0} can be approximated as follows:
\begin{align}
    J_{\phi\phi}
    \approx
    \frac{2P|\gamma|^2}{\sigma^2}
    \sum_{j=1}^{J}\sum_{k=1}^{K}
    \br{\frac{2\pi f_k}{c}\cos\phi}^2
    G_j[k]D_{{\rm eff},j}^2[k],
    \label{eq:JphiphiApprox}
\end{align}
where the definition $G_j[k]\triangleq \|\mathbf g_j[k]\|^2$ was used.
\end{prop}

\begin{proof}
After substituting the AoA derivative $\partial\mathbf a[k]/\partial\phi$ of the ULA steering vector into the FIM, the relevant term becomes a quadratic form in the element positions, i.e., $\diag(x_1,\ldots,x_N)\mathbf a[k]$. Under matched DMA combining, this quadratic form can be interpreted as the variance of the element positions under the combiner-induced power profile, multiplied by the coherent combining gain $G_j[k]$. Hence, the angular information is approximately\footnote{ The FIM term contains $\|\mathbf W_j^{\rm H}[k]\mathbf \diag(x_1,\ldots,x_N)\mathbf a[k]\|^2$ 
that depends on both the amplitudes and phases of the DMA analog combiner across all constituent metamaterials. Under matched combining, the dominant contribution comes from the element-wise power distribution induced by the DMA, while the oscillatory phase-dependent cross terms are assumed to be negligible, or they average out. This allows the exact quadratic form to be approximated by the coherent gain $G_j[k]$ multiplied by the weighted second central moment of the metasurface element positions, where the weights are given by the normalized DMA power profile across the aperture. Thus, the angular information is approximated by the effective aperture variance seen after DMA combining, rather than by the full phase-coupled quadratic form.} governed by the weighted second central moment of the effective aperture after DMA combining. In other words, the AoA sensitivity is determined by how widely the effectively excited DMA elements are distributed around their weighted centroid. 
\end{proof}

Proposition~2 implies that waveguide attenuation affects AoA estimation in two ways: \textit{i}) it reduces the coherent gain \(G_j[k]\); and, through the leakage-induced taper, \textit{ii}) it decreases the effective aperture moment \(D_{{\rm eff},j}^2[k]\) by assigning greater weights to elements proximal to the RF chain, and reduced weights to those located farther along the waveguide. Consequently, the CRB for the AoA estimation is generally more sensitive to waveguide attenuation than a typical analysis based solely on the Signal-to-Noise Ratio (SNR) would suggest. A larger effective aperture produces stronger spatial phase variation across the array, and hence greater angular information, whereas a narrower, or more heavily attenuated, aperture suppresses this variation resulting in the deterioration of the AoA estimation accuracy.

\vspace{-0.2cm}
\subsection{Leakage-Induced CRB Inflation}
In line with \cite{carlson2025wideband}, it is convenient, for interpretive purposes, to isolate a scalar leakage factor. To this end, let ut consider the following attenuation-only profile:
\begin{align}
    h_{{\rm att},n}(f_k)\triangleq\ee^{-\bar{\alpha}_g(f_k)(n-1)d_x},
\end{align}
and then define the normalized leakage efficiency as follows:
\begin{align}
    A_{\rm leak}(f_k)
    \triangleq
    \frac{\abs{\sum_{n=1}^{N}h_{{\rm att},n}(f_k)}^2}
    {N\sum_{n=1}^{N}\abs{h_{{\rm att},n}(f_k)}^2}.
    \label{eq:Aleak}
\end{align}
When the attenuation profile varies slowly across the signal bandwidth, \(A_{\rm leak}(f_k)\) can be approximated as frequency independent, i.e., $A_{\rm leak}(f_k)\approx A_{\rm leak}$ $\forall k=1,\ldots,K$ \cite{carlson2025wideband}.

\begin{prop}
Consider the observation model in \eqref{eq:yjk} and let $\mathbf J_{\rm leak}$ denote the FIM of the unknown channel parameter vector. Assume that, for each DMA phase configuration (state), the effective sensing manifold of the $\ell$-path at the $k$-th SC can be factorized as follows:
\begin{align}
    \mathbf g_{j,\ell}[k]
    \cong
    \sqrt{A_{{\rm leak},j}}\,
    \bar{\mathbf g}_{j,\ell}[k],
    \,\,
    0<A_{{\rm leak},j}\leq 1,
    \label{eq:g_leak_fact}
\end{align}
where $A_{{\rm leak},j}$ is a scalar leakage efficiency factor associated with the $j$-th DMA state, and $\bar{\mathbf g}_{j,\ell}[k]$ denotes the corresponding leakage-free manifold. Let \(\mathbf{J}_j^{(0)}\) denote the contribution of the \(j\)-th DMA state to the FIM in the leakage-free case, namely, for \(A_{{\rm leak},j}=1\). The FIM under leakage is then given by:
\begin{align}
    \mathbf{J}_{\rm leak}
    =
    \sum_{j=1}^{J}
    A_{{\rm leak},j}\,
    \mathbf{J}_j^{(0)}.
    \label{eq:J_leak_sum}
\end{align}
If the same leakage factor applies across all DMA states, i.e., \(A_{{\rm leak},j}=A_{\rm leak}\) \(\forall j=1,\ldots,J\), then \(\mathbf{J}_{\rm leak}=A_{\rm leak}\mathbf{J}_0\) with $\mathbf{J}_0 \triangleq \sum_{j=1}^{J}\mathbf{J}_j^{(0)}$. Hence, provided that \(\mathbf{J}_0\) is nonsingular, the CRB for any scalar parameter \(\theta_u\) can be expresed as follows:
\begin{align}
    \CRB_{\rm leak}(\theta_u)
    =
    \frac{1}{A_{\rm leak}}\,
    \CRB_0(\theta_u),
    \label{eq:crbLeak}
\end{align}
where \(\CRB_{\rm leak}(\theta_u)\triangleq [\mathbf{J}_{\rm leak}^{-1}]_{u,u}\) and
\(\CRB_0(\theta_u)\triangleq [\mathbf{J}_0^{-1}]_{u,u}\) are the CRBs with and without leakage, respectively.
\end{prop}

\begin{proof}
Under \eqref{eq:g_leak_fact}, the mean observation vector at the $j$-th DMA state can be written as $ \boldsymbol{\mu}_{j,k}\cong\sqrt{A_{{\rm leak},j}}\bar{\boldsymbol{\mu}}_{j,k}$, where $\bar{\boldsymbol{\mu}}_{j,k}$ is the corresponding leakage-free mean. Therefore, for any parameter $\theta_u$, $\frac{\partial \boldsymbol{\mu}_{j,k}}{\partial \theta_u}
    \approx
    \sqrt{A_{{\rm leak},j}}\,
    \frac{\partial \bar{\boldsymbol{\mu}}_{j,k}}{\partial \theta_u}$ holds. Substituting this relation into the FIM expression shows that every quadratic score term contributed by state $j$ is multiplied by $A_{{\rm leak},j}$, which directly yields \eqref{eq:J_leak_sum}. If $A_{{\rm leak},j}=A_{\rm leak}$ $\forall j$, then, by factoring out the common scalar, yields $\J_{\rm leak}^{-1}=\frac{1}{A_{\rm leak}}\J_0^{-1}$, and the scalar CRB relation in \eqref{eq:crbLeak} follows immediately.
\end{proof}

Equation \eqref{eq:crbLeak} provides a useful first-order description of leakage-induced CRB inflation. The exact CRB degradation is typically larger because the effective aperture reduction in Proposition 2 acts on top of the scalar gain loss.

\vspace{-0.2cm}
\subsection{Impact on UE Parameter Estimation}
For the direct UE path, the Jacobians of the delay and AoA with respect to the UE position are given as follows:
\begin{align}
    \frac{\partial \tau_0}{\partial \mathbf p_{\rm U}}
    =
    \frac{1}{c}
    \begin{bmatrix}
        \cos\phi_0\\
        \sin\phi_0
    \end{bmatrix},
    \qquad
    \frac{\partial \phi_0}{\partial \mathbf p_{\rm U}}
    =
    \frac{1}{r_0}
    \begin{bmatrix}
        -\sin\phi_0\\
        \cos\phi_0
    \end{bmatrix},
    \label{eq:tauphiJac}
\end{align}
where $r_0\triangleq\norm{\mathbf p_{\rm U}-\mathbf p_{\rm D}}$. Neglecting cross terms for exposition, the UE position EFIM can be written approximately as:
\begin{align}
    \mathbf J_{{\rm p}_{\rm U}}
    \approx
    J_{\tau_0\tau_0}
    \frac{\partial \tau_0}{\partial \mathbf p_{\rm U}}
    \frac{\partial \tau_0}{\partial \mathbf p_{\rm U}^{\rm T}}
    +
    J_{\phi_0\phi_0}
    \frac{\partial \phi_0}{\partial \mathbf p_{\rm U}}
    \frac{\partial \phi_0}{\partial \mathbf p_{\rm U}^{\rm T}}.
    \label{eq:JpApprox}
\end{align}
Evidently, frequency selectivity degrades localization mainly through the reduction of $J_{\tau_0\tau_0}$ via the effective information bandwidth, and through manifold distortion across frequency in $J_{\phi_0\phi_0}$. Additionally, waveguide attenuation degrades localization through coherent gain reduction in both terms, and through effective aperture shrinkage in the angular term.

\begin{table}[!t]
\centering
\caption{The Considered Simulation Parameters.}
\label{tab:sim}
\begin{tabular}{|c|c|c|c|}
\hline
Parameter & Value & Parameter & Value \\
\hline\hline
$f_c$ & 20 GHz & $N$ & 32 \\
\hline
$K$ & 128 & $N_{\rm RF}$, $J$ & 4, 4 \\
\hline
$B=K\Delta f$ & 500 MHz & $Q_c=2\pi f_c/\Gamma$ & 100 \\
\hline
$d_x$ & $\lambda/2$ & $B_{\rm tune}$ & 100--1000 MHz \\
\hline
P & 10 (dBm) & Leakage fraction $\Lambda$ & 0.8 \\
\hline
\end{tabular}
\end{table}

\section{Numerical Results and Discussion}\label{Sec:Results}
In this section, we numerically validate Propositions~1--3 derived for the joint UE localization and SPs sensing system with the wideband DMA response model presented in Section~\ref{Sec:System}. Specifically, we compare the exact FIM-based bounds with their analytical approximations to assess their tightness, as well as to investigate the impact of realistic DMA features on the spatiotemporal parameter estimation performance.

\begin{figure*}[!t]
  \begin{subfigure}[t]{0.33\textwidth}
  \centering
    \includegraphics[width=\textwidth]{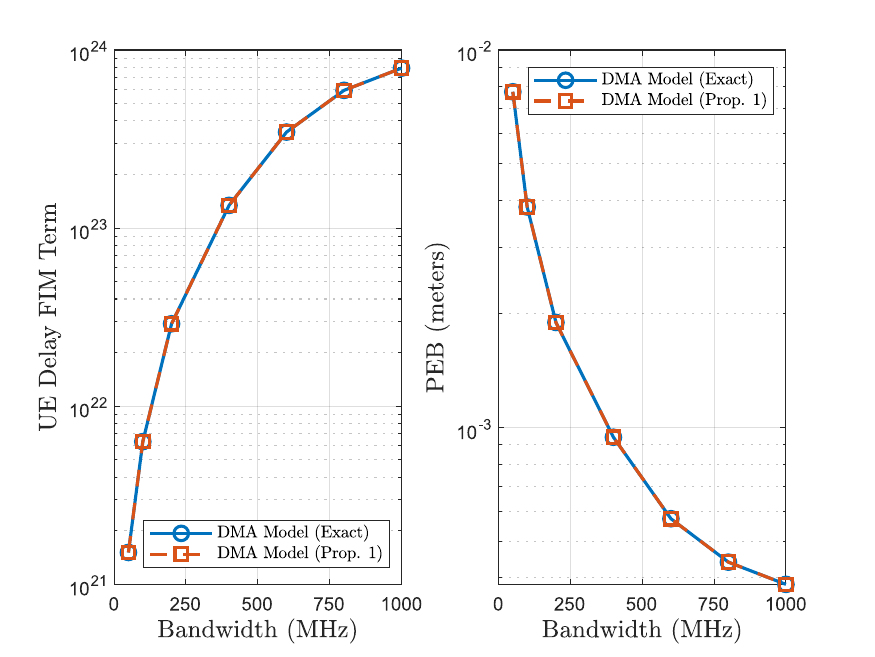}
    \caption{CRB/PEB Delay Inflation.}
    \label{fig:track}
  \end{subfigure}\hfill
  \begin{subfigure}[t]{0.33\textwidth}
  \centering
    \includegraphics[width=\textwidth]{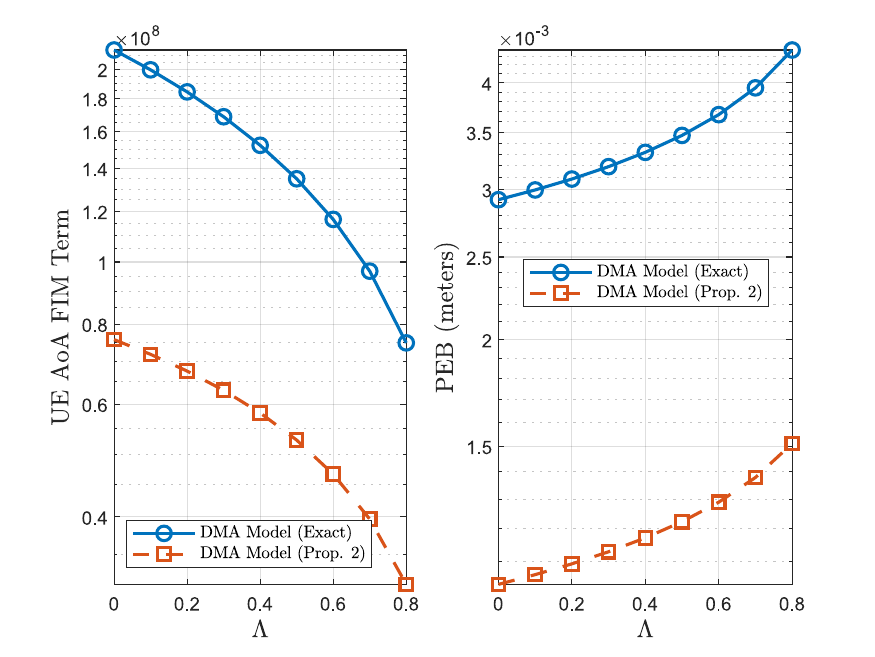}
    \caption{CRB/PEB AoA Inflation.}
    \label{fig:Lower_bound}
  \end{subfigure}\hfill
  \begin{subfigure}[t]{0.33\textwidth}
  \centering
    \includegraphics[width=\textwidth]{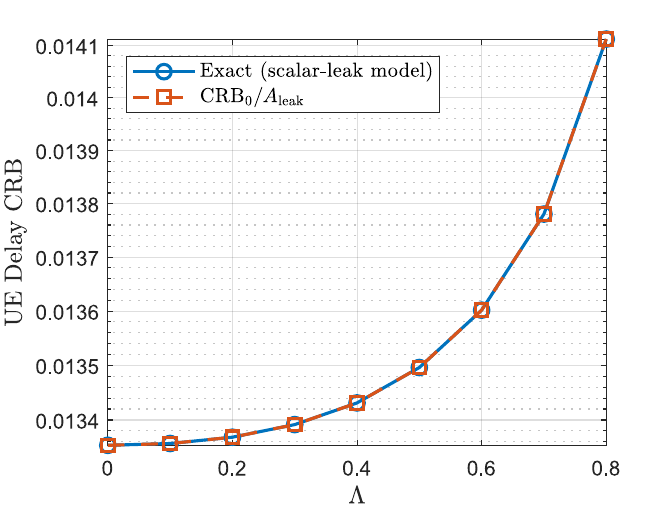}
    \caption{First-Order Leakage Approximation.}
    \label{fig:tradeoff}
  \end{subfigure}
  \caption{\small{Numerical validation of Propositions~1--3 for DMA-based multicarrier sensing. (a) Delay-domain results: exact UE delay FIM term and corresponding PEB, along with the analytical approximations from Proposition~1, versus signal bandwidth $B$. (b) AoA-domain results: exact UE angular-information term and corresponding PEB, together with the effective-aperture-based approximation from Proposition~2, versus the waveguide leakage factor \(\Lambda\). (c) Delay-CRB inflation under waveguide attenuation: exact delay CRB versus the first-order leakage approximation \(\mathrm{CRB}_0/A_{\rm leak}\) from Proposition~3.}}\vspace{-0.4cm}
  \label{fig:Sim}
\end{figure*}

Unless otherwise stated, the simulation setup follows Table~\ref{tab:sim}. The DMA-based RX was centered at the origin and comprises $N$ metamaterial elements arranged along the $x$-axis. The UE and SPs were placed on the two-dimensional plane according to the geometry described in Section~\ref{Sec:System}, with the UE located at $[3,3]$ meters and the $S=2$ SPs located at $[5,3]$ meters and $[4,4]$ meters, respectively. The DMA analog combiner was PEB-optimized in a manner similar to \cite{carlson2025wideband}.

Figure~\ref{fig:Sim} validates the analytical interpretations of Propositions~1--3 under the considered wideband DMA sensing model. As shown in Fig.~\ref{fig:Sim}(a), the approximation in Proposition~1 closely matches both the exact UE delay-information term and the corresponding PEB across the considered bandwidth range. As the bandwidth increases, the delay information improves and the PEB decreases, confirming the expected enhancement in delay resolution and showcasing that this gain is governed by the effective information bandwidth induced by the DMA frequency selectivity. Figure~\ref{fig:Sim}(b) confirms the angular-domain interpretation of Proposition~2 up to multiplicative constants: as the leakage factor \(\Lambda\) increases, the exact angular-information term and the corresponding PEB follow the same trend as the effective-aperture-based approximation. This showcases that waveguide attenuation degrades localization performance through both coherent-gain reduction and effective-aperture shrinkage. Finally, Fig.~\ref{fig:Sim}(c) verifies the first-order leakage law presented in Proposition~3, with the exact delay CRB closely tracking the derived approximation \(\mathrm{CRB}_0/A_{\rm leak}\). Overall, the results confirm that frequency selectivity affects sensing performance primarily through the effective information bandwidth, whereas waveguide attenuation impacts UE localization through both coherent-gain loss and effective-aperture reduction.

\section{Conclusion}
This paper studied multicarrier uplink sensing with a realistic DMA-based RX including multiple RF chains. Exact FIM and EFIM expressions were derived under frequency selective DMA responses, which were subsequently used to characterize the resulting position estimation performance bounds. The analysis showcased that frequency selectivity primarily reduces the effective information bandwidth, whereas waveguide attenuation decreases both the coherent gain and the effective aperture, thereby degrading estimation accuracy. The presented numerical results supported the derived analysis, underscoring the importance of incorporating realistic DMA features into the design of sensing-oriented multicarrier RXs.

\vspace{-0.2cm}
\bibliographystyle{IEEEtran}
\bibliography{ms}

\end{document}